\documentclass[12pt,english]{article}
\usepackage[T1]{fontenc}
\usepackage[latin9]{inputenc}
\usepackage[letterpaper]{geometry}
\usepackage{amsthm}
\usepackage{amsmath}
\usepackage{amssymb}
\usepackage{graphicx}
\usepackage{setspace}
\usepackage[authoryear]{natbib}
\makeatletter
\newcommand{\lyxaddress}[1]{
\par {\raggedright #1
\vspace{1.4em}
\noindent\par}
}
 \theoremstyle{definition}
  \newtheorem{example}{\protect\examplename}
  \theoremstyle{plain}
  \newtheorem{prop}{\protect\propositionname}
\theoremstyle{plain}
\newtheorem{thm}{\protect\theoremname}
  \theoremstyle{plain}
  \newtheorem{lem}{\protect\lemmaname}
  \theoremstyle{remark}
  \newtheorem{rem}{\protect\remarkname}

\DeclareMathOperator{\unif}{U}

\DeclareMathOperator{\pri}{prior}

\makeatother

\usepackage{babel}
  \providecommand{\examplename}{Example}
  \providecommand{\lemmaname}{Lemma}
  \providecommand{\propositionname}{Proposition}
  \providecommand{\remarkname}{Remark}
\providecommand{\theoremname}{Theorem}

\begin{document}

\title{\textbf{\large ~}\\
Blending Bayesian and frequentist methods according to the precision of
prior information with an application to hypothesis testing\\
\textbf{~}}

\maketitle
~~\\
David R. Bickel

\lyxaddress{Ottawa Institute of Systems Biology\\
Department of Biochemistry, Microbiology, and Immunology\\
University of Ottawa; 451 Smyth Road; Ottawa, Ontario, K1H 8M5}
\begin{abstract}
The following zero-sum game between nature and a statistician blends
Bayesian methods with frequentist methods such as p-values and confidence
intervals. Nature chooses a posterior distribution consistent with
a set of possible priors. At the same time, the statistician selects
a parameter distribution for inference with the goal of maximizing
the minimum Kullback-Leibler information gained over a confidence
distribution or other benchmark distribution. An application to testing
a simple null hypothesis leads the statistician to report a posterior
probability of the hypothesis that is informed by both Bayesian and
frequentist methodology, each weighted according how well the prior
is known.

As is generally acknowledged, the Bayesian approach is ideal given
knowledge of a prior distribution that can be interpreted in terms
of relative frequencies. On the other hand, frequentist methods such
as confidence intervals and p-values have the advantage that they
perform well without knowledge of such a distribution of the parameters.
Since neither the Bayesian approach nor the frequentist approach is
entirely satisfactory in situations involving partial knowledge of
the prior distribution, the proposed procedure reduces to a Bayesian
method given complete knowledge of the prior, to a frequentist method
given complete ignorance about the prior, and to a blend between the
two methods given partial knowledge of the prior. The blended approach
resembles the Bayesian method rather than the frequentist method to
the precise extent that the prior is known. 

The problem of testing a point null hypothesis illustrates the proposed
framework. The blended probability that the null hypothesis is true
is equal to the p-value or a lower bound of an unknown Bayesian posterior
probability, whichever is greater. Thus, given total ignorance represented
by a lower bound of 0, the p-value is used instead of any Bayesian
posterior probability. At the opposite extreme of a known prior, the
p-value is ignored. In the intermediate case, the possible Bayesian
posterior probability that is closest to the p-value is used for inference.
Thus, both the Bayesian method and the frequentist method influence
the inferences made.
\end{abstract}

\section{\label{sec:Introduction}Introduction}

\subsection{\label{sub:Motivation}Motivation}

Various compromises between Bayesian and frequentist approaches to
statistical inference represent first attempts to combining attractive
aspects of each approach \citep{good1983good}. While the more recent
the hybrid inference approach of \citet{Yuan20092458} succeeded in
leveraging Bayesian point estimators with maximum likelihood estimates,
reducing to the former or the latter in the presence or absence of
a reliably estimated prior on all parameters, how to extend the theory
beyond point estimation is not yet clear. Further, hybrid inference
in its current form does not cover the case of a parameter of interest
that has a partially known prior. Since such partial knowledge of
a prior occurs in many scientific inference situations, it calls for
a theoretical framework for method development that appropriately
blends Bayesian and frequentist methods.

Ideally, blended inference would meet these criteria:
\begin{enumerate}
\item \textbf{Complete knowledge of the prior.} If the prior is known, the
corresponding posterior is used for inference. Among statisticians,
this principle is almost universally acknowledged. However, it is
rarely the case of the prior is known for all practical purposes.
\item \textbf{Negligible knowledge of the prior.} If there is no reliable
knowledge of a prior, inference is based on methods that do not require
such knowledge. This principle motivates not only the development
of confidence intervals and p-values but also Bayesian posteriors
derived from improper and data-dependent priors. Accordingly, blended
inference must allow the use of such methods when applicable.
\item \textbf{Continuum between extremes.} Inference relies on the prior
to the extent that it is known while relying on the other methods
to the extent that it is not known. Thus, there is a gradation of
methodology between the above two extremes. The premise of this paper
is that this intermediate scenario calls for a careful balance between
pure Bayesian methods on one hand and impure Bayesian or non-Bayesian
methods on the other hand.
\end{enumerate}
Instead of framing the knowledge of a prior in terms of confidence
intervals, as in pure empirical Bayes approaches, it will be framed
more generally herein in terms of a set of plausible priors, as in
interval probability \citep{ISI:000086923900003,ISI:000173555100002,ISI:000224953000004}
and robust Bayesian \citep{berger1984robustness} approaches. Whereas
the concept of an unknown prior cannot arise in strict Bayesian statistics,
it does arise in robust Bayesian statistics when the levels of belief
of an intelligent agent have not been fully assessed \citep{berger1984robustness}.
Unknown priors also occur in many more objective contexts involving
purely frequentist interpretations of probability in terms of variability
in the observable world rather than the uncertainty in the mind of
an agent. For example, frequency-based priors are routinely estimated
under random effects and empirical Bayes models; see, e.g., \citet{efron_large-scale_2010}.
(Remark \ref{rem:interpretation} comments further on interpretations
of probability and relaxes the convenient assumption of a true prior.) 

With respect to the problem at hand, the most relevant robust Bayesian
approaches are the \emph{minimax Bayes risk} ({}``$\Gamma$-minimax'')
practice of minimizing the maximum Bayes risk \citep{Robbins1951131,RefWorks:179,citeulike:7129735}
and the \emph{maxmin expected utility} ({}``conditional $\Gamma$-minimax'')
practice of maximizing the minimum posterior expected payoff or, equivalently,
minimizing the maximum posterior expected loss \citep{ISI:A1989AL55500003,DasGupta1989333,citeulike:7129735,ISI:000173555100002,ISI:000224953000004}.
\citet{ISI:000224953000004} reviews both methods in terms of interval
probabilities that need not be subjective. With typical loss functions,
the former method meets the above criteria for classical minimax alternatives
to Bayesian methods but does not apply to other attractive alternatives.
For example, several confidence intervals, p-values, and objective-Bayes
posteriors routinely used in biostatistics are not minimax optimal.
(\citet{RefWorks:1302} and \citet{FraserAncillaries2004a} argued
that requiring the optimality of frequentist procedures can lead to
trade-offs between hypothetical samples that potentially mislead scientists
or yield pathological procedures.) Optimality in the classical sense is not required of the alternative
procedures under the framework outlined below, which can be understood
in terms of maxmin expected utility with a payoff function that incorporates
the alternative procedures to be used as a benchmark for the Bayesian
posteriors.

\subsection{\label{sec:Overview}Heuristic overview}

To define a general theory of blended inference that meets a formal
statement of the three criteria, Section \ref{sec:General-theory}
introduces a variation of a zero-sum game of \citet{ISI:A1979GQ08900002},
\citet{e3030191}, and \citet{Topsoe2007b}. (The discrete version
of the game also appeared in \citet{Phuber77}, and \citet{Gruenwald20041367}
interpreted it as a special case of the maxmin expected utility problem.)
The {}``nature'' opponent selects a prior consistent with the available
knowledge as the {}``statistician'' player selects a posterior distribution
with the aim of maximizing the minimum information gained relative
to one or more alternative methods. Such benchmark methods may be
confidence interval procedures, frequentist hypothesis tests, or other
techniques that are not necessarily Bayesian.

From that theory, Section \ref{sec:Hypothesis-testing} derives a
widely applicable framework for testing hypotheses. For concreteness,
the motivating results are heuristically summarized here. Consider
the problem of testing $H_{0}:\theta_{\ast}=0$, the hypothesis that
a real-valued parameter $\theta_{\ast}$ of interest is equal to the
point 0 on the real line $\mathbb{R}$. The observed data vector $x$
is modeled as a realization of a random variable denoted by $X$.
Let $p\left(x\right)$ denote the p-value resulting from a statistical
test.

It has long been recognized that the p-value for a simple (point)
null hypothesis is often smaller than Bayesian posterior probabilities
of the hypothesis \citep{Lindley1957b,RefWorks:1007}. Suppose $\theta_{\ast}$
has an unknown prior distribution according to which the prior probability
of $H_{0}$ is $\pi_{0}$. While $\pi_{0}$ is unknown, it is assumed
to be no less than some known lower bound denoted by $\underline{\pi}_{0}$.

Following the methodology of \citet{ISI:A1994QX01100012}, \citet{RefWorks:1218}
found a generally applicable lower bound on the Bayes factor. As Section
\ref{sub:Bound} will explain, that bound immediately leads to
\begin{equation}
\underline{\Pr}\left(H_{0}|p\left(X\right)=p\left(x\right)\right)=\left(1-\left(\frac{1-\underline{\pi}_{0}}{\underline{\pi}_{0}ep\left(x\right)\log p\left(x\right)}\right)\right)^{-1}\label{eq:LFDR-bound}
\end{equation}
as a lower bound on the unknown posterior probability of the null
hypothesis for $p\left(x\right)<1/e$ and to $\underline{\pi}_{0}$
as a lower bound on the probability if $p\left(x\right)\ge1/e$.  

In addition to $\Pr\left(H_{0}|p\left(X\right)=p\left(x\right)\right),$
the unknown Bayesian posterior probability of $H_{0}$, there is a
frequentist posterior probability of $H_{0}$ that will guide selection
of a posterior probability for inference based on $\pi_{0}\ge\underline{\pi}_{0}$
and other constraints summarized by $\Pr\left(H_{0}|p\left(X\right)=p\left(x\right)\right)\ge\underline{\Pr}\left(H_{0}|p\left(X\right)=p\left(x\right)\right)$.
While it is incorrect to interpret the p-value $p\left(x\right)$
as a \emph{Bayesian} probability, it will be seen in Section \ref{sub:Confidence-benchmark}
that $p\left(x\right)$ is a \emph{confidence} posterior probability
that $H_{0}$ is true.

With the confidence posterior as the benchmark, the solution to the
optimization problem described above gives the blended posterior probability
that the null hypothesis is true. It is simply the maximum of the
p-value and the lower bound on the Bayesian posterior probability:
\begin{equation}
\Pr\left(H_{0};p\left(x\right)\right)=p\left(x\right)\vee\underline{\Pr}\left(H_{0}|p\left(X\right)=p\left(x\right)\right).\label{eq:blendedLFDR}
\end{equation}
By plotting $\Pr\left(H_{0};p\left(x\right)\right)$ as a function
of $p\left(x\right)$ and $\underline{\pi}_{0}$, Figures \ref{fig:blend2d}
and \ref{fig:blend3d} illustrate each of the above criteria for blended
inference:
\begin{enumerate}
\item \textbf{Complete knowledge of the prior.} In this example, the prior
is only known when $\underline{\pi}_{0}=1$, in which case 
\[
\Pr\left(H_{0};p\left(x\right)\right)=\underline{\Pr}\left(H_{0}|p\left(X\right)=p\left(x\right)\right)=1
\]
for all $p\left(x\right)$. Thus, the p-value is ignored in the presence
of a known prior. 
\item \textbf{Negligible knowledge of the prior.} There is no knowledge
of the prior when $\underline{\pi}_{0}=0$ and negligible knowledge
when $\underline{\pi}_{0}$ is so low that $\underline{\Pr}\left(H_{0}|p\left(X\right)=p\left(x\right)\right)\le p\left(x\right)$.
In such cases, $\Pr\left(H_{0};p\left(x\right)\right)=p\left(x\right)$,
and the Bayesian posteriors are ignored.
\item \textbf{Continuum between extremes.} When $\underline{\pi}_{0}$ is
of intermediate value in the sense that $\underline{\Pr}\left(H_{0}|p\left(X\right)=p\left(x\right)\right)$
is exclusively between $p\left(x\right)$ and 1, 
\[
\Pr\left(H_{0};p\left(x\right)\right)=\underline{\Pr}\left(H_{0}|p\left(X\right)=p\left(x\right)\right)<1.
\]
Consequently, $\Pr\left(H_{0};p\left(x\right)\right)$ increases gradually
from $p\left(x\right)$ to 1 as $\underline{\pi}_{0}$ increases (Figures
\ref{fig:blend2d} and \ref{fig:blend3d}). In this case, the blended
posterior lies in the set of allowed Bayesian posteriors but is on
the boundary of that set that is the closest to the p-value. Thus,
both the p-value and the Bayesian posteriors influence the blended
posterior and thus the inferences made on its basis.
\end{enumerate}
The plotted parameter distribution will be presented in Section \ref{sub:Blended-posterior}
as a widely applicable blended posterior.

Finally, Section \ref{sec:Remarks} offers additional details and
generalizations in a series of remarks.

\begin{figure}
\includegraphics{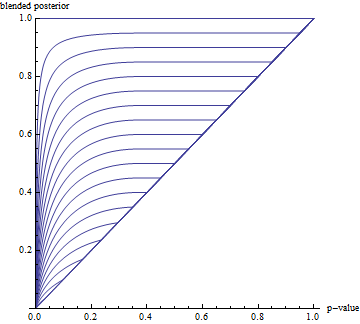}

\caption{Blended posterior probability that the null hypothesis is true versus
the p-value. The curves correspond to lower bounds of prior probabilities
ranging in 5\% increments from 0\% on the bottom to 100\% on the top.\label{fig:blend2d}}
\end{figure}
\begin{figure}
\includegraphics{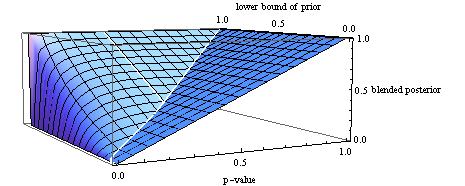}

\includegraphics{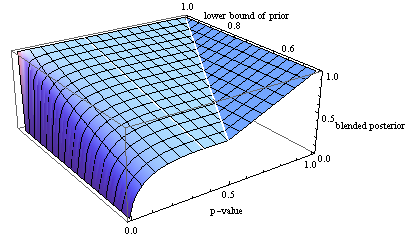}

\caption{Blended posterior probability that the null hypothesis is true versus
the p-value and the lower bound of the prior probability that the
null hypothesis is true. The top plot displays the full domain, half
of which is shown in the bottom plot.\label{fig:blend3d}}
\end{figure}

\section{\label{sec:General-theory}General theory}

\subsection{Preliminary notation and definitions}

Denote the observed data set, typically a vector or matrix of observations,
by $x$, a member of a set $\mathcal{X}$ that is endowed with a $\sigma$-algebra
$\mathfrak{X}$. The value of $x$ determines two sets of posterior
distributions that can be blended for inference about the value of
a target parameter. Much of the following notation is needed to transform
general Bayesian posteriors and confidence posteriors or other benchmark
posteriors such that they are defined on the same measurable space,
that of the target parameter. (A confidence posterior, to be defined
in Section \ref{sub:Confidence-posterior-theory}, is a parameter
distribution from which confidence intervals and p-values may be extracted.
As such, it facilitates blending typical frequentist procedures with
Bayesian procedures.)

\subsubsection{\label{sub:Bayesian-posteriors}Bayesian posteriors}

With some measurable space $\left(\dot{\Theta}_{\ast},\dot{\mathcal{A}}_{\ast}\right)$
for parameter values in $\dot{\Theta}_{\ast}$, let $\mathcal{P}_{\ast}^{\pri}$
denote a set of probability distributions on $\left(\mathcal{X}\times\dot{\Theta}_{\ast},\mathfrak{X}\otimes\dot{\mathcal{A}}_{\ast}\right)$.
Any distribution in $\mathcal{P}_{\ast}^{\pri}$ is called a \emph{prior
(distribution)}, understood in the broad sense of a model that includes
the possible likelihood functions as well as the parameter distribution.
It encodes the constraints and other information available about the
parameter before observing $x$.

On the other hand, any distribution of a parameter is called a \emph{posterior
(distribution)} if it depends on $x$. For some $P_{\ast}^{\pri}\in\mathcal{P}_{\ast}^{\pri}$,
an example of a posterior distribution on $\left(\dot{\Theta}_{\ast},\dot{\mathcal{A}}_{\ast}\right)$
is $\dot{P}_{\ast}=P_{\ast}^{\pri}\left(\bullet\vert X=x\right)$,
where $X$ is a random variable of a distribution on $\left(\mathcal{X},\mathfrak{X}\right)$
that is determined by $P_{\ast}^{\pri}$. $\dot{P}_{\ast}$ is called
a \emph{Bayesian posterior (distribution)} since it is equal to a
conditional distribution of the parameter given $X=x$. Adapting an
apt term from \citet{Topsoe2007b}, the set $\dot{\mathcal{P}}_{\ast}=\left\{ P_{\ast}^{\pri}\left(\bullet\vert X=x\right):P_{\ast}^{\pri}\in\mathcal{P}_{\ast}^{\pri}\right\} $
of Bayesian posteriors on $\left(\dot{\Theta}_{\ast},\dot{\mathcal{A}}_{\ast}\right)$
may be considered the {}``knowledge base.'' For a set $\dot{\Theta}$,
if $\dot{\tau}:\dot{\Theta}_{\ast}\rightarrow\Theta$ is an $\dot{\mathcal{A}}_{\ast}$-measurable
map and if $\dot{\theta}_{\ast}$ has distribution $\dot{P}_{\ast}\in\dot{\mathcal{P}}_{\ast}$,
then $\dot{\theta}=\dot{\tau}\left(\dot{\theta}_{\ast}\right)$, referred
to as an \emph{inferential target of} $\dot{P}_{\ast}$, has induced
probability space $\left(\Theta,\mathcal{A},\dot{P}\right)$. The
set 
\[
\dot{\mathcal{P}}=\left\{ \dot{P}:\dot{\tau}\left(\dot{\theta}_{\ast}\right)\sim\dot{P},\dot{\theta}_{\ast}\sim\dot{P}_{\ast}\in\dot{\mathcal{P}}_{\ast}\right\} 
\]
of all distributions thereby induced and the set $\mathcal{P}$ of
all probability distributions on $\left(\Theta,\mathcal{A}\right)$
are related by $\dot{\mathcal{P}}\subseteq\mathcal{P}$.
\begin{example}
\label{exa:testing}In the hypothesis test of Section \ref{sec:Overview},
$\dot{\theta}=0$ if the null hypothesis that $\dot{\theta}_{\ast}=0$
is true and $\dot{\theta}=1$ if the alternative hypothesis that $\dot{\theta}_{\ast}\ne0$
is true, where $\dot{\theta}_{\ast}$ and $\dot{\theta}$ are random
variables with distributions respectively defined on the Borel space
$\left(\mathbb{R},\mathcal{B}\left(\mathbb{R}\right)\right)$ and
the discrete space $\left(\left\{ 0,1\right\} ,2^{\left\{ 0,1\right\} }\right)$,
where $2^{\left\{ 0,1\right\} }$ is the power set of $\left\{ 0,1\right\} $.
Thus, in this case, $\dot{\tau}$ is the indicator function $1_{\left(-\infty,0\right)\cup\left(0,\infty\right)}:\mathbb{R}\rightarrow\left\{ 0,1\right\} $,
yielding $\dot{\theta}=1_{\left(-\infty,0\right)\cup\left(0,\infty\right)}\left(\dot{\theta}_{\ast}\right)$.
Section \ref{sec:Hypothesis-testing} considers this example in more
detail.
\end{example}
A function that transforms a set of parameter distributions to a single
parameter distribution on the same measurable space is called an \emph{inference
process} \citep{Paris232410,Paris199777}. The resulting distribution
is known as a {}``representation'' \citep{ISI:000173555100002}~or {}``reduction'' \citep{CoherentFrequentism}
of the set. Perhaps the best known inference process for a discrete
parameter set $\Theta$ is that of the \emph{maximum entropy principle},
which would select a member of $\dot{\mathcal{P}}$ such that it has
higher entropy than any other member of the set (see Remark \ref{rem:maxent}).
This paper will propose a wide class of inference processes such that
each transforms $\dot{\mathcal{P}}$ to a member of $\mathcal{P}$
on the basis the following concept of a benchmark distribution on
$\left(\Theta,\mathcal{A}\right)$.

\subsubsection{\label{sub:Benchmark-posteriors}Benchmark posteriors}

For the convenience of the reader, the same Latin and Greek letters
will be used for the set of posteriors that will represent a gold
standard or benchmark method of inference as for the Bayesian posteriors
of Section \ref{sub:Bayesian-posteriors}, with the double-dot $\ddot{\bullet}$
replacing the single-dot $\dot{\bullet}$. Let $\ddot{\mathcal{P}}_{\ast}$
represent a set of posterior distributions on some measurable space
$\left(\ddot{\Theta}_{\ast},\ddot{\mathcal{A}}_{\ast}\right)$, and
let $\ddot{\mathfrak{P}}_{\ast}$ represent a set of such sets. For
instance, considering any $\ddot{P}_{\ast}$ in $\ddot{\mathcal{P}}_{\ast}$,
$\ddot{P}_{\ast}$ may be a confidence posterior (a fiducial-like
distribution to be defined precisely in Section \ref{sub:Confidence-benchmark}),
a generalized fiducial posterior of \citet{RefWorks:1175}, or even
a Bayesian posterior based on an improper prior. (In the first case, nested confidence intervals with inexact coverage
rates generate a set $\ddot{\mathcal{P}}_{\ast}$ of multiple confidence
posteriors rather than the single confidence posterior that is generated
by exact confidence intervals \citep{CoherentFrequentism}.) Suppose there exists a function $\boldsymbol{\ddot{\tau}}:\ddot{\mathfrak{P}}_{\ast}\rightarrow\Theta$
such that $\ddot{P}$, the probability distribution of $\boldsymbol{\ddot{\tau}}\left(\ddot{\mathcal{P}}_{\ast}\right)$,
is defined on $\left(\Theta,\mathcal{A}\right)$. $\ddot{P}$ is
called the \emph{benchmark posterior (distribution)}, and $\ddot{\theta}=\boldsymbol{\ddot{\tau}}\left(\ddot{\mathcal{P}}_{\ast}\right)$
is the \emph{inferential target of }$\ddot{\mathcal{P}}_{\ast}$.
It follows that $\ddot{P}$ is in $\mathcal{P}$ but not necessarily
in $\dot{\mathcal{P}}$. 
\begin{example}
Consider a model in which the full parameter $\dot{\theta}_{\ast}\in\dot{\Theta}_{\ast}$
consists of an interest parameter $\dot{\theta}$ and a nuisance parameter
$\dot{\lambda}$. The measurable space of $\dot{\theta}_{\ast}=\left\langle \dot{\theta},\dot{\lambda}\right\rangle $
is denoted by $\left(\dot{\Theta}_{\ast},\dot{\mathcal{A}}_{\ast}\right)$,
and that of $\dot{\theta}$ by $\left(\Theta,\mathcal{A}\right)$.
Suppose that a set of Bayesian posteriors is available for $\dot{\theta}_{\ast}$
but that nested confidence intervals are only available for an unknown
parameter $\theta\in\Theta$. It follows that a confidence posterior
$\ddot{P}$ is available on $\left(\Theta,\mathcal{A}\right)$ but
not on $\left(\dot{\Theta}_{\ast},\dot{\mathcal{A}}_{\ast}\right)$.
Then the framework of this section can be applied by using the function
$\dot{\tau}$ such that $\theta=\dot{\tau}\left(\dot{\theta}_{\ast}\right)$
in order to project the Bayesian posteriors onto $\left(\Theta,\mathcal{A}\right)$,
the measurable space on which $\ddot{P}$ is defined. In this case,
since there is only one possible benchmark posterior, the function
$\boldsymbol{\ddot{\tau}}$ need not be explicitly constructed.
\end{example}
The function $\boldsymbol{\ddot{\tau}}$ allows consideration of a
set of possible benchmark posteriors by transforming it to a single
benchmark posterior defined on $\left(\Theta,\mathcal{A}\right)$,
the same measurable space as the above Bayesian posteriors of $\dot{\theta}$.
Since that function is unusual, two ways to compose it will now be
explained.
\begin{example}
\label{exa:inference-then-transformation}Consider the inference process
$\ddot{\Pi}:\ddot{\mathfrak{P}}_{\ast}\rightarrow\mathcal{P}_{\ast}$,
where $\mathcal{P}_{\ast}$ is the set of all probability distributions
on $\left(\ddot{\Theta}_{\ast},\ddot{\mathcal{A}}_{\ast}\right)$.
Define the random variable $\ddot{\theta}_{\ast}$ to have distribution
$\ddot{\Pi}\left(\ddot{\mathcal{P}}_{\ast}\right)\left(\bullet\right)=\ddot{\Pi}\left(\ddot{\mathcal{P}}_{\ast}\right)$.
If $\ddot{\tau}:\ddot{\Theta}_{\ast}\rightarrow\Theta$ is an $\ddot{\mathcal{A}}_{\ast}$-measurable
function, then $\ddot{\theta}=\ddot{\tau}\left(\ddot{\theta}_{\ast}\right)$
is the inferential target of $\ddot{\mathcal{P}}_{\ast}$. Further,
the distribution $\ddot{P}$ of $\ddot{\theta}$ is the benchmark
posterior.
\end{example}
~
\begin{example}
\label{exa:transformation-then-inference}Whereas Example \ref{exa:inference-then-transformation}
applied an inference process before a parameter transformation, this
example reverses the order by first applying $\ddot{\tau}$. Let $\ddot{\mathcal{P}}$
denote the subset of $\mathcal{P}$ consisting of all distributions
of the parameters transformed by $\ddot{\tau}$: 
\[
\ddot{\mathcal{P}}=\left\{ P:\ddot{\tau}\left(\ddot{\theta}_{\ast}\right)\sim P,\ddot{\theta}_{\ast}\sim\ddot{P}_{\ast}\in\ddot{\mathcal{P}}_{\ast}\right\} .
\]
Then an inference process transforms $\ddot{\mathcal{P}}$ to the
benchmark posterior $\ddot{P}$, which in turn is the distribution
of $\ddot{\theta}$, the inferential target of $\ddot{\mathcal{P}}_{\ast}$.
\end{example}
~

\subsection{\label{sub:Blended-inference}Blended inference}

In terms of Radon-Nikodym differentiation, the \emph{information divergence}
of $P$ with respect to $Q$ on $\left(\Theta,\mathcal{A}\right)$
is
\begin{equation}
I\left(P||Q\right)=\int dP\log\left(\frac{dP}{dQ}\right)\label{eq:cross-entropy}
\end{equation}
if $P\ll Q$ and $I\left(P||Q\right)=\infty$ otherwise. $I\left(P||Q\right)$
is also known as cross/relative entropy, $I$-divergence, information
for discrimination, and Kullback-Leibler divergence. Other measures
of information may also be used (Remark \ref{rem:generalization}).
For any posteriors $\dot{P}\in\dot{\mathcal{P}}$ and $Q\in\mathcal{P}$,
the \emph{inferential gain} $I\left(\dot{P}||\ddot{P}\rightsquigarrow Q\right)$
\emph{of} $Q$ \emph{relative to} $\ddot{P}$ \emph{given} $\dot{P}$
is the amount of information gained by making inferences on the basis
of $Q$ instead of the benchmark posterior $\ddot{P}$:
\[
I\left(\dot{P}||\ddot{P}\rightsquigarrow Q\right)=I\left(\dot{P}||\ddot{P}\right)-I\left(\dot{P}||Q\right).
\]

Let $\dot{\mathcal{P}}\left(\ddot{P}\right)$ denote the largest
subset of $\dot{\mathcal{P}}$ such that the information divergence
of any of its members with respect to $\ddot{P}$ is finite. That
is,
\begin{equation}
\dot{\mathcal{P}}\left(\ddot{P}\right)=\left\{ \dot{P}\in\dot{\mathcal{P}}:I\left(\dot{P}||\ddot{P}\right)<\infty\right\} ,\label{eq:finite}
\end{equation}
which is nonempty by assumption. (The assumption is not necessary
under the generalization described in Remark \ref{rem:infinity-allowed}.)

The \emph{blended posterior (distribution)} $\hat{P}$ is the probability
distribution on $\left(\Theta,\mathcal{A}\right)$ that maximizes
the inferential gain relative to the benchmark posterior given the
worst-case posterior restricted by the constraints that defined $\dot{\mathcal{P}}$
and $\dot{\mathcal{P}}\left(\ddot{P}\right)$:
\begin{equation}
\inf_{\dot{P}\in\dot{\mathcal{P}}\left(\ddot{P}\right)}I\left(\dot{P}||\ddot{P}\rightsquigarrow\hat{P}\right)=\sup_{Q\in\mathcal{P}}\inf_{\dot{P}\in\dot{\mathcal{P}}\left(\ddot{P}\right)}I\left(\dot{P}||\ddot{P}\rightsquigarrow Q\right),\label{eq:maximin}
\end{equation}
where the supremum and infinum over any set including an indeterminate
number are $\infty$ and $-\infty$, respectively \citep{Topsoe2007b}.
Inferences based on $\hat{P}$ are blended in the sense that they
depend on both $\dot{\mathcal{P}}$ and $\ddot{P}$ in the ways to
be specified in Section \ref{sub:Properties}.

The main result of Theorem 2 of \citet{Topsoe2007b} gives a simply
stated solution of the optimality problem of equation \eqref{eq:maximin}
under broad conditions.
\begin{prop}
\label{pro:minxent}If $I\left(\dot{P}||\ddot{P}\right)<\infty$ for
some $\dot{P}\in\dot{\mathcal{P}}$ and if $\dot{\mathcal{P}}\left(\ddot{P}\right)$
is convex, then the blended posterior $\hat{P}$ is the probability
distribution in $\dot{\mathcal{P}}$ that minimizes the information
divergence with respect to the benchmark posterior:
\begin{equation}
I\left(\hat{P}||\ddot{P}\right)=\inf_{\dot{P}\in\dot{\mathcal{P}}\left(\ddot{P}\right)}I\left(\dot{P}||\ddot{P}\right).\label{eq:minxent-result}
\end{equation}
\end{prop}
\begin{proof}
\citet{Topsoe2007b} proved the result from inequalities of information
theory given the additional stated condition of his Theorem 2 that
$I\left(\dot{P}||\ddot{P}\right)<\infty$ for all $\dot{P}\in\dot{\mathcal{P}}\left(\ddot{P}\right)$.
(See Remark \ref{rem:infinity-allowed}.) The condition that $I\left(\dot{P}||\ddot{P}\right)<\infty$
for some $\dot{P}\in\dot{\mathcal{P}}$ and the above definition of
$\dot{\mathcal{P}}\left(\ddot{P}\right)$ ensure that the condition
is met.
\end{proof}
Alternatively, the minimization of information divergence may define
$\hat{P}$ rather than result from its definition in terms of the
game (Remark \ref{rem:minxent}).

\subsection{\label{sub:Properties}Properties of blended inference}

The desiderata of Section \ref{sec:Introduction} for blended inference
can now be formalized. A posterior distribution $\tilde{P}\left(\bullet;\dot{\mathcal{P}},\ddot{P}\right)$
on $\left(\Theta,\mathcal{A}\right)$ is said to \emph{blend} the
set $\dot{\mathcal{P}}$ of Bayesian posteriors with the benchmark
posterior $\ddot{P}$ for inference about the parameter in $\Theta$
provided that $\tilde{P}\left(\bullet;\dot{\mathcal{P}},\ddot{P}\right)$
satisfies the following criteria under the conditions of Proposition
\ref{pro:minxent}:
\begin{enumerate}
\item \textbf{\label{enu:Complete-knowledge}Complete knowledge of the prior.}
If $\dot{\mathcal{P}}$ has a single member $\dot{P}$, then $\tilde{P}\left(\bullet;\dot{\mathcal{P}},\ddot{P}\right)=\dot{P}$.
\item \textbf{\label{enu:Negligible-knowledge}Negligible knowledge of the
prior.} If $\ddot{P}\in\dot{\mathcal{P}}$ and if $\dot{\mathcal{P}}$
has at least two members, then $\tilde{P}\left(\bullet;\dot{\mathcal{P}},\ddot{P}\right)=\ddot{P}$.
\item \textbf{\label{enu:Continuity-between-extremes}Continuum between
extremes.} For any $D\ge0$ and any $\mathcal{P}^{\star}\subseteq\mathcal{P}$
such that 
\begin{equation}
\sup_{P\in\mathcal{P}^{\star},\dot{P}\in\dot{\mathcal{P}}\left(\ddot{P}\right)}\left|I\left(P||\ddot{P}\right)-I\left(\dot{P}||\ddot{P}\right)\right|\le D\label{eq:contin-condition}
\end{equation}
 and such that $\dot{\mathcal{P}}\left(\ddot{P}\right)\cup\mathcal{P}^{\star}$
is convex, 
\begin{equation}
\left|I\left(\tilde{P}\left(\bullet;\dot{\mathcal{P}}\cup\mathcal{P}^{\star},\ddot{P}\right)||\ddot{P}\right)-I\left(\tilde{P}\left(\bullet;\dot{\mathcal{P}},\ddot{P}\right)||\ddot{P}\right)\right|\le D.\label{eq:contin-criterion}
\end{equation}
\end{enumerate}
\begin{thm}
The blended posterior $\hat{P}$ blends the set $\dot{\mathcal{P}}$
of Bayesian posteriors with the benchmark posterior $\ddot{P}$ for
inference about the parameter in $\Theta$. \end{thm}
\begin{proof}
Since the criteria are only required under the conditions of Proposition
\ref{pro:minxent}, it will suffice to prove that the criteria follow
from equation \eqref{eq:minxent-result}. If $\dot{\mathcal{P}}$
has a single member $\dot{P}$, then equation \eqref{eq:minxent-result}
implies that $\hat{P}=\dot{P}$, thereby ensuring Criterion \ref{enu:Complete-knowledge}.
Similarly, if $\ddot{P}\in\dot{\mathcal{P}}$, then equation \eqref{eq:minxent-result}
implies that $\hat{P}=\ddot{P}$, thus proving that Criterion \ref{enu:Negligible-knowledge}
is met. Assume, contrary to Criterion \ref{enu:Continuity-between-extremes},
that there exist a $D\ge0$ and a $\mathcal{P}^{\star}\subseteq\mathcal{P}$
such that $\dot{\mathcal{P}}\left(\ddot{P}\right)\cup\mathcal{P}^{\star}$
is convex, equation \eqref{eq:contin-condition} is true, and equation
\eqref{eq:contin-criterion} is false with $\tilde{P}\left(\bullet;\dot{\mathcal{P}}\cup\mathcal{P}^{\star},\ddot{P}\right)$
and $\tilde{P}\left(\bullet;\dot{\mathcal{P}},\ddot{P}\right)$ equal
to the blended posteriors respectively using $\dot{\mathcal{P}}\cup\mathcal{P}^{\star}$
and $\dot{\mathcal{P}}$ as the sets of Bayesian posteriors. Then
equation \eqref{eq:minxent-result} can be written as
\[
I\left(\tilde{P}\left(\bullet;\dot{\mathcal{P}}\cup\mathcal{P}^{\star},\ddot{P}\right)||\ddot{P}\right)=\inf_{\dot{P}\in\dot{\mathcal{P}}\left(\ddot{P}\right)\cup\mathcal{P}^{\star}}I\left(\dot{P}||\ddot{P}\right);
\]
\[
I\left(\tilde{P}\left(\bullet;\dot{\mathcal{P}},\ddot{P}\right)||\ddot{P}\right)=\inf_{\dot{P}\in\dot{\mathcal{P}}\left(\ddot{P}\right)}I\left(\dot{P}||\ddot{P}\right).
\]
Hence, with $a\wedge b$ signifying the minimum of $a$ and $b$,
\[
\left|I\left(\tilde{P}\left(\bullet;\dot{\mathcal{P}}\cup\mathcal{P}^{\star},\ddot{P}\right)||\ddot{P}\right)-I\left(\tilde{P}\left(\bullet;\dot{\mathcal{P}},\ddot{P}\right)||\ddot{P}\right)\right|=
\]
\[
\inf_{\dot{P}\in\dot{\mathcal{P}}\left(\ddot{P}\right)}I\left(\dot{P}||\ddot{P}\right)-\inf_{\dot{P}\in\dot{\mathcal{P}}\left(\ddot{P}\right)}I\left(\dot{P}||\ddot{P}\right)\wedge\inf_{P\in\mathcal{P}^{\star}}I\left(P||\ddot{P}\right),
\]
which cannot exceed $\inf_{\dot{P}\in\dot{\mathcal{P}}\left(\ddot{P}\right)}I\left(\dot{P}||\ddot{P}\right)-\inf_{P\in\mathcal{P}^{\star}}I\left(P||\ddot{P}\right)$
and thus, according to equation \eqref{eq:contin-condition}, cannot
exceed $D$. Therefore, the above assumption that equation \eqref{eq:contin-criterion}
is false is contradicted, thereby establishing satisfaction of Criterion
3. \end{proof}
\begin{example}
Suppose the set of possible priors consists of a single frequency-matching
prior, i.e., a prior that leads to 95\% posterior credible intervals
that are equal to 95\% confidence intervals, etc. If the benchmark
posterior is the confidence posterior that yields the same confidence
intervals, then it is the Bayesian posterior distribution corresponding
to the prior. In that case, the blended distribution is equal to that
Bayesian/confidence posterior. Thus, the first condition of blended
inference applies. The second condition would instead apply if the
set of possible priors contained at least one other prior in addition
to the frequency-matching prior. 
\end{example}
Criterion \ref{enu:Continuity-between-extremes} is much stronger
than the heuristic idea of continuity introduced in Section \ref{sub:Motivation}.
Its use of information divergence can be generalized to other measures
of divergence (Remark \ref{rem:generalization}).

\section{\label{sec:Hypothesis-testing}Blended hypothesis testing}

A fertile field of application for the theory of Section \ref{sec:General-theory}
is that of testing hypotheses, as outlined in Section \ref{sec:Overview}.
Building on Example \ref{exa:testing}, this section provides methodology
for a wide class of models used in hypothesis testing.

\subsection{\label{sub:Bound}A bound on the Bayesian posterior}

Defining that class in terms of the concepts of Section \ref{sub:Bayesian-posteriors}
requires additional notation. For a continuous sample space $\mathcal{X}$
and a function $p:\mathcal{X}\rightarrow\left[0,1\right]$ such that
$p\left(X\right)\sim\unif\left(0,1\right)$ under a null hypothesis,
each $p\left(x\right)$ for any $x\in\mathcal{X}$ will be called
a \emph{p-value}. Using some dominating measure, let $f_{0}$ and
$f_{1}$ denote probability density functions of $p\left(X\right)$
under the null hypothesis $\left(\dot{\theta}=0\right)$ and under
the alternative hypothesis $\left(\dot{\theta}=1\right)$, respectively.
For the observed $x$, the likelihood ratio $f_{0}\left(p\left(x\right)\right)/f_{1}\left(p\left(x\right)\right)$
is called the \emph{Bayes factor} since, for a prior distribution
$P_{\ast}^{\pri}$, Bayes's theorem gives
\begin{equation}
\frac{\varphi\left(p\left(x\right)\right)}{1-\varphi\left(p\left(x\right)\right)}=\frac{P_{\ast}^{\pri}\left(\dot{\theta}=0\right)}{P_{\ast}^{\pri}\left(\dot{\theta}=1\right)}\frac{f_{0}\left(p\left(x\right)\right)}{f_{1}\left(p\left(x\right)\right)},\label{eq:posterior-odds}
\end{equation}
where $\varphi\left(p\left(x\right)\right)=P_{\ast}^{\pri}\left(\dot{\theta}=0\vert p\left(X\right)=p\left(x\right)\right)$.
Here, as $\varphi\left(p\left(x\right)\right)$ is a \emph{local false
discovery rate} (LFDR), the letter $\varphi$ abbreviates {}``false''
(\citealp{efron_large-scale_2010}; \citealp{BFDR}). In a parametric
setting, $f_{1}\left(p\left(x\right)\right)$ would be the likelihood
integrated over the prior distribution conditional on the alternative
hypothesis. 

Let $\kappa:\mathcal{X}\rightarrow\mathbb{R}$ denote the function
defined by the transformation $\kappa\left(x\right)=-\log p\left(x\right)$
for all $x\in\mathcal{X}$. Then a probability density of $\kappa\left(X\right)$
under the null hypothesis is the standard exponential density $g_{0}\left(\kappa\left(x\right)\right)=e^{-\kappa\left(x\right)}$.
Assume that, under the alternative hypothesis $\left(\dot{\theta}=1\right)$,
$\kappa\left(X\right)$ admits a density function $g_{1}$ with respect
to the same dominating measure as $g_{0}$. It follows that $g_{0}\left(\kappa\left(x\right)\right)/g_{1}\left(\kappa\left(x\right)\right)=f_{0}\left(p\left(x\right)\right)/f_{1}\left(p\left(x\right)\right)$.
The \emph{hazard rate} $h_{1}\left(\kappa\left(x\right)\right)$ under
the alternative is defined by $h_{1}\left(\kappa\left(x\right)\right)=g_{1}\left(\kappa\left(x\right)\right)/\int_{\kappa\left(x\right)}^{\infty}g_{1}\left(k\right)dk$
for all $x\in\mathcal{X}$, and $h_{1}:\left(0,\infty\right)\rightarrow\left[0,\infty\right)$
is called the \emph{hazard rate function}. 

\citet{RefWorks:1218} obtained the following lower bound $\underline{b}\left(p\left(x\right)\right)$
of the Bayes factor $b\left(x\right)$.
\begin{lem}
If $h_{1}$ is nonincreasing, then, for all $x\in\mathcal{X}$, 
\begin{equation}
b\left(p\left(x\right)\right)=\frac{f_{0}\left(p\left(x\right)\right)}{f_{1}\left(p\left(x\right)\right)}\ge\underline{b}\left(p\left(x\right)\right)=\begin{cases}
-ep\left(x\right)\log p\left(x\right) & \text{if }p\left(x\right)<1/e;\\
1 & \text{if }p\left(x\right)\ge1/e.
\end{cases}\label{eq:BF-bound}
\end{equation}

\end{lem}
The condition on the hazard rate defines a wide class of models that
is useful for testing simple null hypotheses. A broad subclass will
now be defined by imposing constraints on $\pi_{0}=P_{\ast}^{\pri}\left(\dot{\theta}=0\right),$
the prior probability that the null hypothesis is true, in addition
to the hazard rate condition. Specifically, $\pi_{0}$ is known to
have $\underline{\pi}_{0}\in\left[0,1\right]$ as a lower bound. Thus,
rearranging equation \eqref{eq:posterior-odds} as
\[
\varphi\left(p\left(x\right)\right)=\left(1+\left(\frac{1-\pi_{0}}{\pi_{0}b_{0}\left(p\left(x\right)\right)}\right)\right)^{-1},
\]
a lower bound on $\varphi\left(p\left(x\right)\right)$ is 
\[
\underline{\Pr}\left(H_{0}|p\left(X\right)=p\left(x\right)\right)=\underline{\varphi}\left(p\left(x\right)\right)=\left(1+\left(\frac{1-\underline{\pi}_{0}}{\underline{\pi}_{0}\underline{b}\left(p\left(x\right)\right)}\right)\right)^{-1},
\]
leading to equation \eqref{eq:LFDR-bound}.

Let $\mathcal{P}$ consist of all probability distributions on $\left(\Theta,\mathcal{A}\right)=\left(\left\{ 0,1\right\} ,2^{\left\{ 0,1\right\} }\right)$.
The subset $\dot{\mathcal{P}}$ consists of all $\dot{P}\in\mathcal{P}$
such that $\dot{P}\left(\dot{\theta}=0\right)\ge\underline{\varphi}\left(p\left(x\right)\right)$.

\subsection{\label{sub:Confidence-benchmark}A confidence benchmark posterior}

\subsubsection{\label{sub:Confidence-posterior-theory}Confidence posterior theory}

The following parametric framework facilitates the application of
 Section \ref{sub:Benchmark-posteriors} to hypothesis testing. The
observation $x$ is an outcome of the random variable $\ddot{X}$
of probability space $\left(\mathcal{X},\mathfrak{X},P_{\theta_{\ast},\lambda_{\ast}}\right)$,
where the interest parameter $\theta_{\ast}\in\ddot{\Theta}_{\ast}$
and a nuisance parameter $\lambda_{\ast}$ (in some set $\ddot{\Lambda}_{\ast}$)
are unknown. Let $S:\ddot{\Theta}_{\ast}\times\mathcal{X}\rightarrow\left[0,1\right]$
and $t:\mathcal{X}\times\ddot{\Theta}_{\ast}\rightarrow\mathbb{R}$
denote functions such that $S\left(\bullet;x\right)$ is a distribution
function, $S\left(\theta_{\ast};X\right)\sim\unif\left(0,1\right)$,
and
\[
S\left(\theta_{\ast};x\right)=P_{\theta_{\ast},\lambda_{\ast}}\left(t\left(\ddot{X};\theta_{\ast}\right)\ge t\left(x;\theta_{\ast}\right)\right)
\]
for all $x\in\mathcal{X}$, $\theta_{\ast}\in\ddot{\Theta}_{\ast}$,
and $\lambda_{\ast}\in\ddot{\Lambda}_{\ast}$. $S$ is known as a
\emph{significance function}, and $t$ as a \emph{pivot }or test statistic.
It follows that $p\left(x\right)=S\left(0;x\right)$ is a p-value
for testing the hypothesis that $\theta_{\ast}=0$ and that $\left[S^{-1}\left(\alpha;X\right),S^{-1}\left(\beta;X\right)\right]$
is a $\left(\beta-\alpha\right)100\%$ confidence interval for $\theta_{\ast}$
given any $\alpha\in\left[0,1\right]$ and $\beta\in\left[\alpha,1\right]$.
Thus, whether a significance function is found from p-values over
a set of simple null hypotheses or instead from a set of nested confidence
intervals, it contains the information needed to derive either (\citealp{RefWorks:127};
\citealp{RefWorks:1037}; \citealp{CoherentFrequentism,conditional2009}).

Let $\ddot{\theta}_{\ast}$ denote the random variable of the probability
measure $\ddot{P}_{\ast}$ that has $S\left(\bullet;x\right)$ as
its distribution function. In other words, $\ddot{P}_{\ast}\left(\ddot{\theta}_{\ast}\le\theta_{\ast}\right)=S\left(\theta_{\ast};x\right)$
for all $\theta_{\ast}\in\ddot{\Theta}_{\ast}$. $\ddot{P}_{\ast}$
is called a \emph{confidence posterior (distribution)} since it equates
the frequentist coverage rate of a confidence interval with the probability
that the parameter lies in the fixed, observed confidence interval:
\begin{eqnarray*}
\beta-\alpha & = & P_{\theta_{\ast},\lambda_{\ast}}\left(\theta_{\ast}\in\left[S^{-1}\left(\alpha;X\right),S^{-1}\left(\beta;X\right)\right]\right)\\
 & = & \ddot{P}_{\ast}\left(\ddot{\theta}_{\ast}\in\left[S^{-1}\left(\alpha;x\right),S^{-1}\left(\beta;x\right)\right]\right)
\end{eqnarray*}
for all $x\in\mathcal{X}$, $\theta_{\ast}\in\ddot{\Theta}_{\ast}$,
and $\lambda_{\ast}\in\ddot{\Lambda}_{\ast}$. The term {}``confidence
posterior'' ~\citep{CoherentFrequentism,conditional2009} is preferred
here over the usual term {}``confidence distribution'' \citep{RefWorks:127}
to emphasize its use as an alternative to Bayesian posterior distributions.
\citet{Polansky2007b}, \citet{RefWorks:1037}, and \citet{CoherentFrequentism}
provide generalizations to vector parameters of interest. Extensions
based on multiple comparison procedures are sketched in Remark \ref{rem:MCP}.

\subsubsection{A confidence posterior for testing}

For the application to two-sided testing of a simple null hypothesis,
let $\theta_{\ast}=\left|\theta_{\ast\ast}\right|$, the absolute
value of a real parameter $\theta_{\ast\ast}$ of interest, leading
to $\ddot{\Theta}_{\ast}=\left[0,\infty\right)$. Then $p\left(x\right)=S\left(0;x\right)$
is equivalent to a two-tailed p-value for testing the hypothesis that
$\theta_{\ast\ast}=0$. Since $\ddot{P}_{\ast}\left(\ddot{\theta}_{\ast}\le0\right)=S\left(0;x\right)$
and since $\ddot{P}_{\ast}\left(\ddot{\theta}_{\ast}\le0\right)=\ddot{P}_{\ast}\left(\ddot{\theta}_{\ast}=0\right)$,
it follows that $p\left(x\right)=\ddot{P}_{\ast}\left(\ddot{\theta}_{\ast}=0\right)$,
i.e., the p-value is equal to the probability that the mull hypothesis
is true. 

If $\ddot{P}_{\ast}$ is the only confidence posterior under consideration,
then $\ddot{\mathcal{P}}_{\ast}=\left\{ \ddot{P}_{\ast}\right\} $,
and there is no need for an inference process. Following the terminology
of Example \ref{exa:inference-then-transformation}, $\ddot{\tau}:\ddot{\Theta}_{\ast}\rightarrow\Theta$
is defined by $\ddot{\tau}\left(\ddot{\theta}_{\ast}\right)=1_{\left(0,\infty\right)}\left(\ddot{\theta}_{\ast}\right)$.
By implication, $\ddot{\theta}=0$ if $\ddot{\theta}_{\ast}=0$ and
$\ddot{\theta}=1$ if $\ddot{\theta}_{\ast}>0$. Thus, $p\left(x\right)=\ddot{P}_{\ast}\left(\ddot{\theta}_{\ast}=0\right)$
ensures that $\ddot{P}\left(\ddot{\theta}=0\right)=p\left(x\right)$,
which in turn implies $\ddot{P}\left(\ddot{\theta}=1\right)=1-p\left(x\right)$.
\begin{example}
In the various $t$-tests, $\theta_{\ast}$ is the mean of $X$ or
a difference in means, and the statistic $t\left(X;0\right)$ is the
absolute value of a statistic with a Student $t$ distribution of
known degrees of freedom. The above formalism then gives the usual
two-sided p-value from a $t$-test as $\ddot{P}\left(\ddot{\theta}=0\right)$
and $p\left(x\right)$. Specials cases of this $\ddot{P}$ have been
presented as fiducial distributions (\citet{RefWorks:1369};\citealp{smallScale}).
\end{example}

\subsection{\label{sub:Blended-posterior}A blended posterior for testing}

This subsection blends the above set $\dot{\mathcal{P}}$ of Bayesian
posteriors with the above confidence posterior $\ddot{P}$ as prescribed
by Section \ref{sub:Blended-inference}. Gathering the results of
Sections \ref{sub:Bound} and \ref{sub:Confidence-benchmark}, 
\[
\dot{\mathcal{P}}=\left\{ \dot{P}\in\mathcal{P}:\dot{P}\left(\dot{\theta}=0\right)\ge\underline{\varphi}\left(p\left(x\right)\right)\right\} ;
\]
\[
\ddot{P}\left(\ddot{\theta}=0\right)=p\left(x\right)=1-\ddot{P}\left(\ddot{\theta}=1\right).
\]
Equation \eqref{eq:finite} then implies that
\[
\dot{\mathcal{P}}\left(\ddot{P}\right)=\left\{ \dot{P}\in\mathcal{P}:\underline{\varphi}\left(p\left(x\right)\right)\le\dot{P}\left(\dot{\theta}=0\right)<1\right\} ,
\]
in which the first inequality is strict if and only if $\underline{\varphi}\left(p\left(x\right)\right)=0$
and the second inequality is strict unless $p\left(x\right)=1$. Since
$\dot{\mathcal{P}}\left(\ddot{P}\right)$ is convex, Proposition \ref{pro:minxent}
yields
\begin{equation}
\hat{P}\left(\theta=0\right)=\begin{cases}
\underline{\varphi}\left(p\left(x\right)\right) & \text{if }p\left(x\right)<\underline{\varphi}\left(p\left(x\right)\right)\\
p\left(x\right) & \text{if }p\left(x\right)\ge\underline{\varphi}\left(p\left(x\right)\right)
\end{cases},\label{eq:blended-probability}
\end{equation}
where $\theta$ is the random variable of distribution $\hat{P}$.
With the identities $\underline{\varphi}\left(p\left(x\right)\right)=\underline{\Pr}\left(H_{0}|p\left(X\right)=p\left(x\right)\right)$
and $\hat{P}\left(\theta=0\right)=\Pr\left(H_{0};p\left(x\right)\right)$
and with the establishment of equation \eqref{eq:LFDR-bound} by Section
\ref{sub:Bound}, equation \eqref{eq:blended-probability} verifies
the claim of equation \eqref{eq:blendedLFDR} made in Section \ref{sec:Overview}.

\section{\label{sec:Remarks}Remarks}
\begin{rem}
\label{rem:interpretation}As mentioned in Section \ref{sub:Motivation},
the use of Bayes's theorem with proper priors need not involve subjective
interpretations of probability. The set of posteriors may be determined
by interval constraints on the corresponding priors without any requirement
that they model levels of belief \citep{ISI:000086923900003,ISI:000173555100002,ISI:000224953000004}.
However, subjective applications of blended inference are also possible.
While the framework was developed with an unknown prior in mind, the
concept of imprecise or indeterminate probability \citep{RefWorks:Walley1991}
could take the place of the set in which an unknown prior lies. By
allowing the partial order of agent preferences, imprecise probability
theories need not assume the existence of any true prior \citep{RefWorks:Walley1991,ColettiScozzafava2002b}.
As often happens, the same mathematical framework is subject to very
different philosophical interpretations. 
\end{rem}
~
\begin{rem}
\label{rem:maxent}Technically, the principle of maximum entropy \citep{Paris232410,Paris199777}
mentioned in Section \ref{sub:Bayesian-posteriors} could be used
if $\Theta$ is finite or countable infinite. However, unlike the
proposed methodology, that practice is equivalent to making the benchmark
posterior $\ddot{P}$ depend on the function $\dot{\tau}$ that maps
a parameter space to $\Theta$ rather than on a method of data analysis
that is coherent in the sense that its posterior depends on the data
rather than on the hypothesis. If blending with such a method is not
desired, one may average the Bayesian posteriors with respect to some
measure that is not a function of $\Theta$. For example, averaging
with respect to the Lebesgue measure, as \citet{CoherentFrequentism} did with confidence posteriors,
 leads to $\left(1+\underline{\varphi}\left(p\left(x\right)\right)\right)/2$
as the posterior probability of the null hypothesis under the assumptions
of Section \ref{sub:Bound}. Remark \ref{rem:minxent} discusses a
more tenable version of the maximum entropy principle for blended
inference.
\end{rem}
~
\begin{rem}
\label{rem:generalization}Using definitions of divergence that include
information divergence \eqref{eq:cross-entropy} as a special case,
\citet{Gruenwald20041367} and \citet{ISI:000222924800003} generalized
variations of Proposition \ref{pro:minxent}. The theory of blended
inference extends accordingly. 
\end{rem}
~
\begin{rem}
\label{rem:infinity-allowed}A generalization of Section \ref{sec:General-theory}
in a different direction from that of Remark \ref{rem:generalization}
replaces each {}``$\inf_{\dot{P}\in\dot{\mathcal{P}}\left(\ddot{P}\right)}$''
of equation \eqref{eq:maximin} with {}``$\inf_{\dot{P}\in\dot{\mathcal{P}}}$.''
For that optimization problem, Theorem 2 of \citet{Topsoe2007b}
has the condition that $\dot{P}\in\dot{\mathcal{P}}\implies I\left(\dot{P}||\ddot{P}\right)<\infty$
in addition to the convexity of $\dot{\mathcal{P}}$ that Proposition
\ref{pro:minxent} of the present paper requires. Thus, in that formulation,
the blended posterior $\hat{P}$ need not satisfy equation \eqref{eq:minxent-result}
even if $\dot{\mathcal{P}}$ is convex.
\end{rem}
~
\begin{rem}
\label{rem:minxent}A posterior distribution $\hat{P}$ that \emph{is
defined by} 
\begin{equation}
I\left(\hat{P}||\ddot{P}\right)=\inf_{\dot{P}\in\dot{\mathcal{P}}}I\left(\dot{P}||\ddot{P}\right)\label{eq:minxent}
\end{equation}
satisfies the desiderata of Section \ref{sub:Properties} whether
or not the conditions of Proposition \ref{pro:minxent} hold. While
certain axiomatic systems \citep[e.g.,][]{Csiszar19912032} lead to
this generalization of the principle of maximum entropy (Remark \ref{rem:maxent}),
the optimization problem of equation \eqref{eq:maximin} seems more
compelling in this context and defines $\hat{P}$ even when no distribution
satisfying equation \eqref{eq:minxent} exists.
\end{rem}
~
\begin{rem}
\label{rem:MCP}In the presence of multiple comparisons, the confidence
posteriors of Section \ref{sub:Confidence-posterior-theory} may be
adjusted to control a family-wise error rate or false coverage rate
\citep{RefWorks:1276}, if desired. Either error rate would then take
the place of the conventional confidence level as the confidence posterior
probability.
\end{rem}
~
\section*{Acknowledgments}

This research was partially supported  by the Canada Foundation for
Innovation, by the Ministry of Research and Innovation of Ontario,
and by the Faculty of Medicine of the University of Ottawa. 

\begin{flushleft}
\bibliographystyle{elsarticle-harv}
\bibliography{refman}

\par\end{flushleft}

\newpage{}

\end{document}